\documentclass[11pt, letterpaper]{article}

\usepackage[utf8x]{inputenc}
\usepackage{fullpage}
\usepackage{amssymb, amsmath}
\usepackage{amsthm}
\usepackage{graphicx}
\usepackage{algorithm}
\usepackage{algorithmic}
\usepackage{xspace}
\usepackage{enumitem}
\usepackage{tabularx}

\usepackage{typearea}
\typearea{15}

\newenvironment{properties}[1]{
	\begin{enumerate}[label={\normalfont (\ref{#1}\alph*)},labelwidth=*,leftmargin=*,itemsep=0pt,topsep=3pt]
}{
	\end{enumerate}
}

\DeclareMathOperator{\E}{\mathbb E}
\DeclareMathOperator{\union}{\bigcup}

\newcommand{\R}{\mathbb{R}}
\newcommand{\Z}{\mathbb{Z}}

\newtheorem{theorem}{Theorem}[section]
\newtheorem{lemma}[theorem]{Lemma}

\newtheorem{corollary}[theorem]{Corollary}
\newtheorem{definition}[theorem]{Definition}
\newtheorem{claim}[theorem]{Claim}

\newcommand{\set}[1]{\left\{#1\right\}}
\newcommand{\cardinal}[1]{\left|#1\right|}
\newcommand{\floor}[1]{\left\lfloor#1\right\rfloor}
\newcommand{\ceil}[1]{\left\lceil#1\right\rceil}
\newcommand{\calJ}{{\mathcal J}}
\newcommand{\calS}{{\mathcal S}}
\newcommand{\calT}{{\mathcal T}}

\newcommand{\eps}{\epsilon}

\newcommand{\tcalS}{{\widetilde \calS}}
\newcommand{\tC}{{\widetilde C}}

\newcommand{\dav}{d_{\mathsf {av}}}
\newcommand{\LP}{\mathsf{LP}}

\newcommand{\KM}{{\sf KM}\xspace}
\newcommand{\CKM}{{\sf CKM}\xspace}

\makeatletter
\def\thm@space@setup{%
  \thm@preskip=5pt \thm@postskip=4pt
}
\makeatother

\makeatletter
\renewcommand{\paragraph}{%
  \@startsection{paragraph}{4}%
  {\z@}{1ex \@plus 1ex \@minus .2ex}{-1em}%
  {\normalfont\normalsize\bfseries}%
}
\makeatother

\AtBeginDocument{
	\setlength{\abovedisplayskip}{3pt}
	\setlength{\belowdisplayskip}{3pt}
}

\begin{document}


\title{Approximating capacitated $k$-median with $(1+\eps)k$ open facilities}
\author{Shi Li\thanks{TTIC, shili@ttic.edu}}
\date{}
\maketitle 
\begin{abstract}

In the capacitated $k$-median (\CKM) problem, we are given a set $F$ of facilities, each facility $i \in F$ with a capacity $u_i$,  a set $C$ of clients, a  metric $d$ over $F \cup C$ and an integer $k$.  The goal is to open $k$ facilities in $F$ and connect the clients $C$ to the open facilities such that each facility $i$ is connected by at most $u_i$ clients, so as to minimize the total connection cost.  

In this paper, we give the first constant approximation for \CKM, that only violates the cardinality constraint by a factor of $1+\eps$.  This generalizes the result of [Li15], which only works for the uniform capacitated case.   Moreover, the approximation ratio we obtain is $O\big(\frac{1}{\eps^2}\log\frac1\eps\big)$, which is an exponential improvement over the ratio of $\exp(O(1/\epsilon^2))$ in [Li15].   The natural LP relaxation for the problem, which almost all previous algorithms for \CKM are based on, has unbounded integrality gap even if $(2-\eps)k$ facilities can be opened.  We introduce a novel configuration LP for the problem, that overcomes this integrality gap.

\end{abstract}
\thispagestyle{empty}

\newpage
\setcounter{page}{1}

\section{Introduction}

In the capacitated $k$-median (\CKM) problem, we are given a set $F$ of facilities, each facility $i$ with a capacity $u_i \in \Z_{>0}$,  a set $C$ of clients, a  metric $d$ over $F \cup C$ and an upper bound $k$ on the number of facilities we can open.  The goal is to find a set $S \subseteq F$ of at most $k$ open facilities and a connection assignment $\sigma : C \to S$ of clients to open facilities such that $\cardinal{\sigma^{-1}(i)} \leq u_i$ for every facility $i \in S$, so as to minimize the connection cost $\sum_{j \in C}d(j, \sigma(j))$.  

A special case of the problem is the classic $k$-median (\KM) problem, in which all facilities $i \in F$ have $u_i = \infty$. There has been a steady stream of papers on approximating \KM \cite{AGK01, CG99, CGT99, JMS02, JV01, LS13}.  The current best approximation ratio for the problem is $2.611+\eps$ due to Byrka et al.\ \cite{BPR15}.  On the negative side, it is NP-hard to approximate the problem within a factor of $1+2/e-\eps\approx 1.736$ \cite{JMS02}. 

Very little is known about the \CKM problem.  Most previous algorithms \cite{ABG14, BFR15, CGT99, CR05, Li14} are based on the natural LP relaxation for the problem, which has unbounded integrality gap even if all capacities are the same, and one is allowed to violate the cardinality constraint (the constraint that at most $k$ facilities can be open) \emph{or} the capacity constraints by a factor of $2-\eps$. As a result, these algorithms can only give pseudo-approximation algorithms for \CKM.  Indeed, for algorithms based on the natural LP relaxation, near-optimal cardinality violation and capacity violation constants have been obtained by Gijswijt and Li \cite{ABG14} and \cite{Li14}, respectively.  \cite{ABG14} gave a $(7+\eps)$-approximation for \CKM that opens $2k$ facilities. \cite{Li14} gave an $O(1/\eps)$-approximation with $2+\eps$ factor violation on the capacity constraints.  The first algorithm that breaks the factor of 2 barrier on the cardinality violation factor is given by Li \cite{Li15}. \cite{Li15} gave an $\exp(O(1/\eps^2))$-approximation algorithm with $(1+\eps)$-violation in the number of open facilities for the uniform \CKM, the special case of \CKM when all facilities $i$ have the same capacity $u_i = u$. His algorithm is based on a stronger LP called ``rectangle LP''. 

There are two slightly different versions of \CKM. The version as we described is the hard \CKM problem, where each facility can be opened once; while in the soft \CKM problem, multiple copies of each facility can be open. Hard \CKM is more general as one can convert a soft \CKM instance to a hard \CKM instance by making $k$ copies of each facility. The result of \cite{CR05} is for soft \CKM while the other mentioned results are for hard \CKM. In particular, \cite{Li15} showed that when the capacities are the same, hard \CKM and soft \CKM are equivalent, up to a constant loss in the approximation ratio. 


\paragraph{Our contributions} We generalize the result of \cite{Li15} to the non-uniform \CKM problem. Moreover, we improve the approximation ratio exponentially from $\exp(O(1/\eps^2))$ in \cite{Li15} to $O\big(\frac{1}{\eps^2}\log\frac1\eps\big)$.  To be more precise, we give an $\Big(1+\eps, O\big(\frac{1}{\eps^2}\log\frac1\eps\big)\Big)$-approximation algorithm for soft \CKM with running time $n^{O(1/\eps)}$, where an $(\alpha, \beta)$-approximation algorithm is an algorithm which outputs a solution with at most $\alpha k$ open facilities whose cost is at most $\beta$ times the optimum cost with $k$ open facilities. 
\begin{theorem}
\label{thm:main} 
Given a soft capacitated \CKM instance, and a parameter $\eps > 0$, we can find in $n^{O(1/\eps)}$ time a solution to the instance with at most $(1+\eps)k$ open facilities whose cost is at most $O\big(\frac{1}{\eps^2}\log\frac1\eps\big)$ times the cost of optimum solution with $k$ open facilities. 
\end{theorem}

As most previous results are based on the basic LP relaxation for the problem, which has unbounded integrality gap even if we can open $(2-\eps)k$ facilities, our result gives the first constant approximation algorithm for \CKM which only violates the cardinality constraint by $1+\eps$ factor.  Compared to the result of \cite{Li15}, our algorithm works for general capacities and gives an exponentially better approximation ratio. On the downside, the running time of our algorithm is $n^{O(1/\eps)}$, instead of $n^{O(1)}$. It only works for soft-capacitated $k$-median: we may open $2$ copies of the same facility due to some technicality.

\paragraph{Overview} Our algorithm is based on a novel configuration LP; we highlight the key ideas here.  The following bad situation is a barrier for the basic LP relaxation. There is an isolated group of facilities and clients in the metric. The fractional solution uses $y\geq 1$ fractional open facilities to serve all clients in the group. However, the integral solution needs to use at least $\ceil{y}$ facilities. $\ceil{y}/y$ can be as large as $2-\eps$.   To bound the number of open facilities by $(1+\eps)k$, we need to handle the bad case when $y \leq \ell_1$ for some large enough $\ell_1 = \Theta(1/\eps)$. For the uniform capacitated case, this was handled by the ``rectangle LP'' in \cite{Li15}. However, the LP heavily depends on the uniformity of the capacities and generalizing it to the non-uniform case seems very hard.

Instead, we use a more complicated configuration LP. Let $B$ be the set of facilities in the group. If we know that the optimum solution opens at most $\ell_1$ facilities in $B$, we can afford to use $n^{O(\ell_1)}$ configuration variables and natural constraints to characterize the convex hull of all valid integral solutions w.r.t $B$. For example, for any $S \subseteq B$ of size at most $\ell$, we can use $z^B_{S}$ to indicate the event that the set of open facilities in $B$ is exactly $S$.

But we do not know if at most $\ell_1$ facilities are open in $B$. To overcome this issue, we also have a variable $z^B_\bot$ to indicate the event that number of open facilities in $B$ is more than $\ell_1$.  When we condition on $z^B_\bot = 0$,  we obtain a convex hull of integral solutions w.r.t $B$. If we condition on $z^B_\bot = 1$, we only get a fractional solution w.r.t $B$ to the basic LP. However, this is good enough, as the bad situation corresponds to $z^B_\bot = 0$.

We have constraints for every subset $B \subseteq F$ in our configuration LP.  As a result, the LP is hard to solve. We use a standard trick that has been used in many previous papers (e.g. \cite{ASS14, Car00, Li15}).  Given a fractional solution, our rounding algorithm either finds a good integral solution, or finds a violated constraint. This allows us to combine our rounding algorithm with the ellipsoid method.


Our rounding algorithm uses the framework of \cite{Li15}. We create a clustering of facilities, each cluster with a center that is called a representative. Then we partition the set of representatives into groups, and bound the number of open facilities and the connection cost via a group-by-group analysis. The improved approximation ratio comes from a novel partitioning algorithm, where we partition the representatives based on the minimum spanning tree over these representatives.


\paragraph{Organization} In Section~\ref{sec:config-LP}, we introduce the basic and the configuration LP relaxations for \CKM. Then in Section~\ref{sec:basic-rounding}, we give an $(O(1), O(1))$-approximation for \CKM based on the basic LP. The result is not new and the constants are worse than those in \cite{ABG14}.  However, the algorithm serves as a starting line for our $\Big(1 + \eps, O\big(\frac1{\epsilon^2}\log\frac1\eps\big)\Big)$-approximation algorithm.  Then, in Section~\ref{sec:rounding}, we give the rounding algorithm based on the configuration LP. We end this paper with some open problems in Section~\ref{sec:discussion}. All omitted proofs can be found in Appendix~\ref{appendix:proofs}.
\section{The basic LP and the configuration LP}
\label{sec:config-LP}

In this section, we give our configuration LP for \CKM. We start with the following basic LP relaxation:
\begin{equation}
\textstyle \min \qquad \sum_{i \in F, j\in C}d(i,j)x_{i,j} \qquad \text{s.t.} \tag{Basic LP} \label{LP:basic}
\end{equation}
\vspace*{-18pt}

\begin{minipage}{0.4\textwidth}
\begin{alignat}{2}\setlength{\abovedisplayskip}{0pt}\setlength{\belowdisplayskip}{0pt}
\textstyle \sum_{i \in F} y_i &\leq  k;  \label{LPC:k-facilities} \\
\textstyle  \sum_{i \in F}x_{i,j} &=1, &\qquad &\forall j \in C; \label{LPC:client-must-connect} \\
\textstyle  x_{i,j} &\leq y_i, &\qquad &\forall i \in F, j \in C; \label{LPC:connect-to-open}
\end{alignat}
\end{minipage}
\begin{minipage}{0.55\textwidth}
\begin{alignat}{2}\setlength{\abovedisplayskip}{0pt}\setlength{\belowdisplayskip}{0pt}
\textstyle  \sum_{j \in C}x_{i,j} &\leq u_iy_i, &\qquad &\forall i \in F; \label{LPC:capacity} \\
\textstyle  0 \leq x_{i,j}, y_i &\leq 1, &\qquad &\forall i \in F, j \in C. \label{LPC:xy-non-neg} \\[8pt]
 \nonumber
\end{alignat}
\end{minipage}
\vspace*{1pt}

In the above LP, $y_i$ indicates whether facility $i$ is open or not, and $x_{i, j}$ indicates whether the client $j$ is connected to facility $i$.  Constraint~(\ref{LPC:k-facilities}) restricts us to open at most $k$ facilities, Constraint~(\ref{LPC:client-must-connect}) requires every client to be connected to a facility, Constraint~(\ref{LPC:connect-to-open}) says that a client can only be connected to an open facility and Constraint~(\ref{LPC:capacity}) is the capacity constraint. In the integer programming, we require $x_{i, j}, y_i \in \set{0, 1}$ for every $i \in F, j \in C$. In the LP relaxation, we relax the constraint to $x_{i, j}, y_i \in [0, 1]$.  In the soft-capacitated case, $y_i$ may be any non-negative integer. We can make $k$ copies of each facility $i$ and assume the instance is hard-capacitated.  Thus we assume $y_i \in \set{0,1}$ in a valid integral solution.  

The above LP has unbounded integrality gap, even when all facilities have the same capacity $u$ and we are allowed to open $(2-\eps)k$ facilities.  In the gap instance we have $u$ separate groups, each containing $u+1$ clients and $2$ facilities, and we are allowed to open $k = u + 1$ facilities. The fractional solution can open $1+1/u$ facilities in each group and use them to serve the $(1+1/u)u = u+1$ clients in the group. (This can be achieved by, for example, setting $y_i = 1/2 + 1/(2u)$ for each of the two facilities $i$ and $x_{i,j} = 1/2$ every facility $i$ and client $j$ in the group.)  In an integral solution, there is a group in which at most one facility is open, even if we are allowed to open $2u - 1 = 2k-3$ facilities. Some client in this group must be connected to a facility outside the group. Thus, if the the groups are far away from each other, the integrality gap is unbounded. 

The gap instance suggested the following bad situation.  There is an isolate group $B$ of facilities, with some nearby clients.  The fractional solution uses $y_B$ open facilities to serve the nearby clients; however serving these clients integrally requires at least $\ceil{y_B}$ facilities. So if $y_B$ is a small integer plus a tiny fractional number, then $\ceil{y_B}/y_B > 1+\epsilon$.  So, we can not afford to open $\ceil{y_B}$ facilities inside $B$.

This motivates the following idea to handle the isolated group $B$.  Let $\ell_1 = \Theta(1/\eps)$ be large enough. If $y_B \geq \ell_1$, then we will have $\ceil{y_B}/y_B \leq 1+\epsilon$. So, the bad case happens only if  $y_B \leq \ell_1$.  If we know that the optimum solution opens at most $\ell_1$ facilities in $B$, then we can afford to have configuration variables such as $z^B_{S, i,j}$, for every $S\subseteq B, |S| \leq \ell_1, i \in S, j \in C$, to indicate the event that the set of open facilities in $B$ is exactly $S$, and $j$ is connected to $i$.  With some natural constraints on these variables, we can overcome the gap instance. Since we do not know if the optimum solution opens at most $\ell_1$ facilities in $B$,  we shall allow $S$ to take value ``$\bot$'', which means that $S$ has size more than $\ell_1$. In this case, we do not care what is the set exactly; knowing that the size is more than $\ell_1$ is enough. 

\smallskip
Now we formally define the configuration LP.  Let us fix a set $B\subseteq F$ of facilities.  Let $\calS = \set{S \subseteq B: |S| \leq \ell_1}$ and $\tcalS = \calS \cup \set{\bot}$, where $\bot$ stands for ``any subset of $B$ with size more than $\ell_1$''; for convenience, we also treat $\bot$ as a set such that $i \in \bot$ holds for every $i \in B$.   For $S \in \calS$, let $z^B_S$ indicate the event that the set of open facilities in $B$ is exactly $S$ and $z^B_\bot$ indicate the event that the number of open facilities in $B$ is more than $\ell_1$. 

For every $S \in \tcalS$ and $i \in S$, $z^B_{S, i}$ indicates the event that $z^B_S = 1$ and $i$ is open. (If $i \in B$ but $i \notin S$, then the event will not happen.)  Notice that when $i \in S \neq \bot$, we always have $z^B_{S, i} = z^B_S$; we keep both variables to simplify the description of the configuration LP.  For every $S \in \tcalS, i \in S$ and client $j \in C$,  $z^B_{S, i, j}$ indicates the event that $z^B_{S, i} = 1$ and $j$ is connected to $i$.  In an integral solution, all the above variables are $\set{0,1}$ variables.  The following constraints are valid. To help understand the constraints, it is good to think of $z^B_{S, i}$ as $z^B_S \cdot y_i$ and $z^B_{S, i, j}$ as $z^B_S \cdot x_{i, j}$.
\vspace*{-15pt}

\noindent\begin{minipage}{0.52\textwidth}
\begin{align}
\sum_{S \in \tcalS }z^B_S &= 1; \label{CLPC:add-to-one} \\
\sum_{S \in \tcalS:i \in S}z^B_{S, i} &= y_i, & \forall& i \in B; \label{CLPC:add-to-y} \\
\sum_{S \in \tcalS : i \in S}z^B_{S, i, j} &= x_{i,j}, & \forall& i \in B, j \in C; \label{CLPC:add-to-x} \\
0 \leq z^B_{S, i, j} \leq z^B_{S, i} &\leq z^B_S, & \forall & S \in \tcalS, i \in S, j \in C; \label{CLPC:non-negative}
\end{align}
\end{minipage}
\begin{minipage}{0.48\textwidth}
\vspace*{20pt}
\begin{align}
z^B_{S, i}  &= z^B_S, &\forall& S \in \calS, i \in S; \label{CLPC:i-irrelevant} \\
\sum_{i \in S}z^B_{S, i, j} &\leq z^B_S, &\forall& S \in \tcalS, j \in C; \label{CLPC:j-connection-bound}\\
\sum_{j \in C}z^B_{S, i, j} &\leq u_i z^B_{S, i}, & \forall & S \in \tcalS, i \in S; \label{CLPC:capacity} \\
\sum_{i \in B}z^B_{\bot, i} &\geq \ell_1 z^B_\bot. \label{CLPC:more-than-ell-facilities} \\
\nonumber
\end{align}
\end{minipage}
\vspace*{-12pt}

Constraint~\eqref{CLPC:add-to-one} says that $z^B_S = 1$ for exactly one $S \in \tcalS$. Constraint~\eqref{CLPC:add-to-y} says that if $i$ is open then there is exactly one $S \in \tcalS$ such that $z^B_{S, i} = 1$.  Constraint~\eqref{CLPC:add-to-x} says that if $j$ is connected to $i$ then there is exactly one $S \in \tcalS$ such that $z^B_{S, i, j} = 1$. Constraint~\eqref{CLPC:non-negative} is by the definition of the variables. Constraint~\eqref{CLPC:i-irrelevant} holds as we mentioned earlier.  Constraint~\eqref{CLPC:j-connection-bound} says that if $z^B_S = 1$ then $j$ can be connected to at most 1 facility in $S$. Constraint~\eqref{CLPC:capacity} is the capacity constraint. Finally, Constraint~\eqref{CLPC:more-than-ell-facilities} says that if $z^B_{\bot} = 1$, then at least $\ell_1$ facilities in $B$ are open.   


If $z^B_\bot = 0$, the above polytope is integral;  this can be seen from the integrality of matching polytopes. This is not true if $z^B_\bot > 0$.  The variables $\set{z^B_{\bot, i}/z^B_\bot}_{i \in B}$ and $\set{z^B_{\bot, i, j}/z^B_\bot}_{i \in B, j \in C}$ only define a fractional solution to the basic LP for the instance defined by $B$ and $C$ (not all clients in $C$ need to be connected). This is sufficient as the bad case happens only if a few facilities are open in $B$. 

Our configuration LP is obtained from the basic LP, by adding the $z$ variables and Constraints~\eqref{CLPC:add-to-one} to~\eqref{CLPC:more-than-ell-facilities} for every $B \subseteq F$.  Fixing a set $B \subseteq F$ and $(x, y)$, we can check in time $n^{O(1/\eps)}$ if there are $z$ variables satisfying Constraints~\eqref{CLPC:add-to-one} to~\eqref{CLPC:more-than-ell-facilities}, since the total number of variables and constraints is $n^{O(\ell_1) = O(1/\eps)}$.   As there are exponential number of sets $B$, we do not know how to solve this LP.  

Instead, we transform the above configuration LP so that it contains only $x$ and $y$ variables.  The new LP will have infinite number of constraints. Fix a subset $B \subseteq F$. Constraints~\eqref{CLPC:add-to-one} to~\eqref{CLPC:more-than-ell-facilities} can be written as $Mz\succeq b + M'x + M''y$, where $M, M', M''$ are some matrices, $x, y, z$ are the column vectors containing all $x, y, z$ variables respectively, and $b$ is a column vector\footnote{We break an equality into two inequalities.}. By the duality of linear programming, the set of constraints is feasible for a fixed $(x, y)$, iff for every vector $g$ such that $g^TM = 0$, we have $g^T(b + M'x + M''y) \leq 0$. Thus, we can convert Constraints~\eqref{CLPC:add-to-one} to~\eqref{CLPC:more-than-ell-facilities} in the following way: for every $g$ such that $g^TM = 0$, we have the constraint $g^T(b + M'x + M''y) \leq 0$. All these constraints are linear in $x$ and $y$ variables.  Given a fixed $(x, y)$ for which the system defined by constraints~\eqref{CLPC:add-to-one} to~\eqref{CLPC:more-than-ell-facilities} is infeasible, we can find a vector $g$  such that $g^TM = 0$ and $g^T(b + M'x + M''y) > 0$.

Thus, our final configuration LP contains only $x, y$ variables, but infinite number of constraints\footnote{Since the polytope defined by $x, y, z$ variables has finite number of facets, so does its projection to $x, y$ coordinates. Thus, only a finite number of constraints matter. However, it is hard to define this finite set; and it is not important.}.  We can apply the standard trick, which has been used in, e.g, \cite{ASS14} and \cite{Li15}. Given a fractional solution $(x, y)$ to the basic LP relaxation, our rounding algorithm either constructs an integral solution with the desired properties, or outputs a set $B \subseteq F$ such that Constraints~\eqref{CLPC:add-to-one} to~\eqref{CLPC:more-than-ell-facilities} are infeasible.  In the latter case, we can find a violated constraint. This allows us to run the ellipsoid method.

\section{An $(O(1), O(1))$-approximation based on the basic LP relaxation}
\label{sec:basic-rounding}

In this section,  we describe an $(O(1), O(1))$-approximation for \CKM based on the basic LP relaxation.  This result is not new and our constants are worse than those in \cite{ABG14}. The purpose of section is only to set up a starting line for our $\Big(1+\eps, O\big(\frac{1}{\eps^2}\log\frac1\eps\big)\Big)$-approximation; most of the components in this section will be used in our new algorithm. 

After the set of open facilities is decided, the optimum connection assignment from clients to facilities can be computed by solving a minimum cost $b$-matching instance. Due to the integrality of the matching polytope, we may allow the connections to be fractional. That is, if there is a good fractional assignment, then there is a good integral assignment. So we can use the following framework.  Initially there is one unit of demand at each client $j \in C$. During our algorithm, we move demands fractionally inside $F \cup C$, incurring a cost of $xd(i, j)$ for moving $x$ units of demand from $i$ to $j$. At the end, all the demands are moved to $F$. If a facility $i \in F$ has $\alpha_i$ units of demand, we open $\ceil{\alpha_i/u_i}$ copies of $i$ to satisfy the demand. Our goal is to bound the total moving cost and the number of open facilities.  

A standard approach to facility location problems is to partition the facilities into many clusters. Each cluster contains a set of nearby facilities and the fractional number of open facilities in each cluster is not too small.  Each cluster is associated with a center $v \in C$ which we call \emph{client representatives}.

Focus on a fractional solution $(x, y)$ to the basic LP. Let $\dav(j) = \sum_{i \in F}x_{i, j}d(i, j)$ be the connection cost of $j$, for every client $j \in C$. Then the value of the solution $(x, y)$ is $\LP := \sum_{i \in F, j \in C}x_{i,j}d(i,j)=\sum_{j \in C}\dav(j)$.   For any set $F' \subseteq F$ of facilities and any set $C' \subseteq C$ of clients, we shall let $x_{F', C'} := \sum_{i \in F', j \in C'} x_{i, j}$; we simply write $x_{i, C'}$ for $x_{\set{i}, C'}$ and  $x_{F', j}$ for $x_{F', \set{j}}$. For any $F' \subseteq F$, let $y_{F'} := y(F') := \sum_{i \in F'}y_i$.

We shall use $R$ to denote the set of client representatives. Let $R = \emptyset$ initially. Repeat the following process until $C$ becomes empty. We select the client $v \in C$ with the smallest $\dav(v)$ and add it to $R$. We remove all clients $j$ such that $d(j, v) \leq 4 \dav(j)$ from $C$ (thus, $v$ itself is removed). We shall use $v$ and its derivatives to index representatives, and $j$ and its derivatives to index general clients. 

We partition the set $F$ of locations according to their nearest representatives in $R$. Let $U_v = \emptyset$ for every $v \in R$ initially. For each location $i \in F$, we add $i$ to $U_v$ for the $v \in R$ that is closest to $i$. Thus, $\set{U_v:v \in R}$ forms a Voronoi diagram of $F$ with $R$ being the centers. For any subset $V \subseteq R$ of representatives, we use $U_V := U(V) := \union_{v \in V} U_v$ to denote the union of Voronoi regions with centers $V$.

\begin{claim} \label{claim:representatives}
The following statements hold:
\begin{properties}{claim:representatives}
\item for all $v, v' \in R, v \neq v'$, we have $d(v, v') > 4 \max\set{\dav(v), \dav(v')}$; \label{property:representatives-far-away}
\item for all $j \in C$, there exists $v \in R$, such that $\dav(v) \leq \dav(j)$ and $d(v, j) \leq 4\dav(j)$; \label{property:near-a-representative}
\item $y(U_v) \geq 1/2$ for every $v \in R$; \label{property:bundle-large} 
\item for any $v \in R$, $i \in U_v$ and $j \in C$, we have $d(i, v) \leq d(i, j) + 4\dav(j)$. \label{property:facility-to-representative}
\end{properties}
\end{claim}

We move all demands to client representatives. First, for each client $j \in C$ and $i \in F$, we move $x_{i,j}$ units of demand from $j$ to $i$. Obviously, the moving cost is exactly $\LP$.   After this step, each $i \in F$ has $\sum_{j \in C}x_{i,j} = x_{i,C} \leq u_iy_i$ units of demand.   Second, for each $v \in R$ and $i \in U_v$, we move the $x_{i,C}$ units of demand at $i$ to $v$. Then, all demands are moved to $R$ and a representative $v \in R$ has $\sum_{i \in U_v} x_{i,C} = x_{U_v, C}$ units of demand.  Corollary~\ref{cor:moving-to-representatives} bounds the moving cost for the second step.

\begin{definition}
Let $D_i := \sum_{j \in C}x_{i,j}d(i,j)$ and $D'_i := \sum_{j \in C}x_{i,j}\dav(j)$ for every $i \in F$. Let $D_S := D(S) := \sum_{i \in S}D_i$ and $D'_S := D'(S) := \sum_{i \in S}D'_i$ for every $S \subseteq F$.  Obviously $D_F = D'_F = \LP$.
\end{definition}
Thus, we can think of $D_i$ (similarly, $D'_i$) as a distribution of the total cost $\LP$ to the facilities. We shall use $D$ and $D'$ to charge the cost of integral solutions. 

\begin{lemma} 
\label{lemma:moving-to-representatives}
For every $v \in R$, we have
$\sum_{i \in U_v}x_{i,C}d(i,v) \leq D(U_v) + 4D'(U_v)$.
\end{lemma}
\begin{proof}
By Property~\ref{property:facility-to-representative}, we have $d(i, v) \leq d(i, j) + 4\dav(j)$ for every $i \in U_v$ and $j \in C$.  Thus, 
{
\begin{align*}
\textstyle \sum_{i \in U_v}x_{i,C} d(i,v) \leq \sum_{i \in U_v, j \in C}x_{i,j} \big(d(i, j) + 4\dav(j)\big) = \sum_{i \in U_v}(D_i + 4D'_i) = D(U_v) + 4D'(U_v). &\qedhere
\end{align*}
}
\end{proof}

Adding the above inequality for all $v \in R$ gives the following corollary.
\begin{corollary}
\label{cor:moving-to-representatives}
$\sum_{v \in R, i \in U_v}x_{i,C} d(i, v) \leq D_F + 4D'_F = 5\LP$.
\end{corollary}

So, moving demands from $F$ to their respective centers incurs a cost of at most $5\LP$.  Focus on a representative $v \in R$; it has $x_{U_v, C}$ units of demand after the moving. For each $i \in U_v$, we shall move $\alpha_i$ units of demand from $v$ to $i$. Consider the following LP with variables $\set{\alpha_i}_{i \in U_v}$: minimize $\sum_{i \in U_v} \alpha_i d(i,v)$ subject to $\sum_{i \in U_v} \alpha_i / u_i \leq y(U_v), \sum_{i \in U_v} \alpha_i = x_{U_v, C}$ and $\alpha_i \in [0, u_i]$ for every $i \in U_v$.  The value of the LP is at most $\sum_{i \in U_v} x_{i,C} d(i,v)$ which can be achieved by setting $\alpha_i = x_{i,C}$ for every $i \in U_v$. We select a vertex solution $\set{\alpha^*_i}_{i \in U_v}$ for the LP. Then all but at most 2 facilities $i \in U_v$ have $\alpha^*_i \in \set{0, u_i}$. Each $i \in U_v$ gets $\alpha^*_i$ units of demands from $v$; the moving cost is $\sum_{i \in U_v}\alpha^*_i d(i,v) \leq \sum_{i \in U_v} x_{i, C} d(i,v) \leq D(U_v) + 4D'(U_v)$ by Lemma~\ref{lemma:moving-to-representatives}. We open $\sum_{i \in U_v} \ceil{\alpha^*_i/u_i}\leq \big\lfloor \sum_{i \in U_v} \alpha^*_i/u_i + 2 \big\rfloor \leq \big\lfloor y(U_v)\big\rfloor + 2$ facilities in $U_v$. The total moving cost over all $v \in R$ in this step is $\sum_{v \in R}(D(U_v) + 4D'(U_v)) = 5\LP$. Since $y(U_v) \geq 1/2$ for every $v \in R$ by Property~\ref{property:bundle-large}, we have $\frac{\floor{y(U_v)}+2}{y(U_v)} \leq 4$, implying that at most $4k$ facilities are open. This gives a $(4, 11)$-approximation for \CKM. 

\section{Rounding a fractional solution to the configuration LP}
\label{sec:rounding}

In this section, we show how to round a fractional solution to the configuration LP to obtain our $\Big(1+\epsilon, O\big(\frac{1}{\eps^2}\log\frac1\eps\big)\Big)$-approximation. To be more accurate, given a fractional solution $(x, y)$ to the basic LP, the rounding algorithm either succeeds, or finds a set $B \subseteq F$ for which Constraints~\eqref{CLPC:add-to-one} to~\eqref{CLPC:more-than-ell-facilities} is infeasible.  $\dav(j), \LP, x_{F', C'}, y_{F'} = y(F'), D$ and $D'$ are defined in the same way as in Section~\ref{sec:basic-rounding}.  We also construct the set $R$ of client representatives and the clustering $\set{U_v}_{v \in R}$ as in Section~\ref{sec:basic-rounding}.  We let $d(A, B):= \min_{i \in A, j \in B} d(i, j)$ denote the minimum distance between $A$ and $B$, for any $A, B \subseteq F \cup C$; we simply use $d(i, B)$ for $d(\set{i}, B)$.   In this section, $\ell = \Theta(1/\eps)$ is a large enough integer, whose value will be given later, $\ell_1 = 2\ell + 2$ and $\ell_2 = \Theta(\ell \log \ell)$ is large enough, given later by Lemma~\ref{lemma:dealing-with-concentrated-sets}.

We give an overview of the algorithm. The rounding algorithm in Section~\ref{sec:basic-rounding} opens $\floor{y(U_v)}+2$ facilities inside $U_v$.  If $y(U_v) \geq \Theta(1/\eps)$ for every $v \in R$, then the algorithm gives $O(1)$-approximation with $(1+\eps)k$ open facilities.  But we only have $y(U_v) \geq 1/2$. In order to save the number of opening facilities, we follow the framework of \cite{Li15} that combines the representatives to form bigger groups.  We move demands within each group. In each group, we open $2$ more facilities than the number given by the fractional solution.  To obtain the improved approximation ratio, we use a novel process to partition $R$ into groups, based on coloring the edges of the minimum spanning tree of the metric $(R, d)$.  

If there are no so-called ``concentrated sets'' in a group, then the moving cost within the group can be charged locally using the $D$ and $D'$ values. For a concentrated set $J$, we need Constraints~\eqref{CLPC:add-to-one} to~\eqref{CLPC:more-than-ell-facilities} to hold for the set $B = U_J$. We pre-open a set of facilities in $U_J$ and pre-assign a set of clients to these facilities based on the values of the $z$ variables.   After the pre-assignment, the moving cost for the remaining demands can be charged locally. We now describe the algorithm in more detail.

\newcommand{\MST}{\mathsf{MST}}

\paragraph{Partition $R$ into groups} To partition the set $R$ of representatives into groups, we run the classic Kruskal's algorithm to find the minimum spanning tree $\MST$ of the metric $(R, d)$, and then color the edges in $\MST$ using black, grey and white colors. In Kruskal's algorithm, we maintain the set $E_\MST$ of edges added to $\MST$ so far and a partition $\calJ$ of $R$ into groups.  Initially, we have $E_\MST = \emptyset$ and $\calJ = \{\{v\}:v \in R\}$. The length of an edge $e \in {R \choose 2}$ is the distance between the two endpoints of $e$. We sort all edges in $R \choose 2$ according to their lengths, breaking ties arbitrarily. For each pair $(u, v)$ in this order, if $u$ and $v$ are not in the same group in $\calJ$, we add the edge $(u, v)$ to $E_\MST$ and merge the two groups containing $u$ and $v$ respectively.



We now color the edges in $E_\MST$.  For every $v \in R$, we say the weight of $v$ is $y(U_v)$; so every representative $v \in R$ has weight at least $1/2$ by Property~\ref{property:bundle-large}. For a subset $J \subseteq R$ of representatives, we say $J$ is large if the weight of $J$ is at least $\ell$, i.e, $y(U_J) \geq \ell$; we say $J$ is small otherwise.  For any edge $e = (u, v)\in E_\MST$, we consider the iteration in Kruskal's algorithm in which the edge $e$ is added to $\MST$. After the iteration we merged the group $J_u$ containing $u$ and the group $J_v$ containing $v$ into a new group $J_u \cup J_v$.  If both $J_u$ and $J_v$ are small, then we call $e$ a black edge.  If $J_u$ is small and $J_v$ is big, we call $e$ a grey edge, directed from $u$ to $v$; similarly, if $J_v$ is small and $J_u$ is big, $e$ is a grey edge directed from $v$ to $u$. If both $J_u$ and $J_v$ are big, we say $e$ is a white edge. So, we treat black and white edges as undirected edges and grey edges as directed edges. 

We define a black component of $\MST$ to be a maximal set of vertices connected by black edges.   The following claim is straightforward. It follows from the fact that a black component $J \subseteq R$ appeared as a group at some iteration of Kruskal's algorithm for computing $\MST$. 

\begin{claim}
	\label{claim:black-component-in-set}
	Let $J$ be a black component of $\MST$. Then for every black edge $(u, v)$ in $J \choose 2$, we have $d(u, v) \leq d(J, R \setminus J)$.
\end{claim}

We contract all the black edges in $\MST$ and remove all the white edges. The resulting graph is a forest $\Upsilon$ of trees.  Each node  (we use the word ``nodes'' for vertices in the contracted graph) in $\Upsilon$ is correspondent to a black component, and each edge is a directed grey edge.  For every node $p$ in $\Upsilon$, we use $J_p \subseteq R$ to denote the black component correspondent to the node $p$.  Abusing notations slightly, we define $U_p:=U(J_p) = \union_{v \in J_p} U_v$. The weight of the node $p$ is the total weight of representatives in $J_p$, i.e, $y(U_p)$.

\begin{lemma}
	\label{lemma:contracted-tree}
	For any tree $\tau \in \Upsilon$, the following statements are true: 
	\begin{properties}{lemma:contracted-tree}
		\item $\tau$ has a root node $r_\tau$ such that all grey edges in $\tau$ are directed towards $r_\tau$;
		\label{property:rooted-tree}
		\item $J_{r_\tau}$ is big and $J_p$ is small for all other nodes $p$ in $\tau$;
		\label{property:nodes-in-tree-small}
		\item in any leaf-to-root path of $\tau$, the lengths of grey edges form a non-increasing sequence;
		\label{property:lengths-decrease}
		\item for any non-root node $p$ of $\tau$, the length of the grey edge in $\tau$ connecting $p$ to its parent is exactly $d(J_p, R \setminus J_p)$;
		\label{property:grey-edge-length}
		\item for any non-root node $p$ of $\tau$, the length of any black edge in ${J_p \choose 2}$ is at most $d(J_p, R \setminus J_p)$.
		\label{property:black-edge-length}
	\end{properties}
\end{lemma}

We now break $\tau$ into a set $\calT_\tau$ of edge-disjoint sub-trees using the following greedy algorithm. Consider the deepest node $p$ in $\tau$ such that the total weight of all descendant nodes of $p$ (we do not count the weight of $p$) is at least $\ell$.  First assume that the node $p$ exists.   

Notice that the weight of each node other than the root node $r_\tau$ is less than $\ell$ by Property~\ref{property:nodes-in-tree-small}.  By our choice of $p$, the weight of any sub-tree rooted at a child of $p$ is at most $\ell + \ell = 2\ell$. Thus, a simple greedy algorithm can give us a collection of sub-trees rooted at some children of $p$, with total weight between $\ell$ and $2\ell$.   We build a tree $T$ as follows:  take the collection of sub-trees, the node $p$, as well as the edges connecting the roots of the sub-trees to $p$. We add the tree $T$ into $\calT_\tau$. Let $r_T = p$ be the root node of $T$.  

We remove the collection of sub-trees from $\tau$ and repeat the above process to find another tree $T$. We terminate the process when the node $p$ does not exist. In this case, we add the remaining tree to $\calT_\tau$.  It is easy to see the following statements about $\calT_\tau$.
\begin{lemma}
	\label{lemma:decomposition}
	\begin{properties}{lemma:decomposition}
		\item Every non-root node $p$ of $\tau$ appears in exactly one tree in $\calT_\tau$ as a non-root. \label{property:every-node-appears-once}
		\item The number of trees in $\calT_\tau$ is at most $1/\ell$ times the total weight of nodes in $\tau$.
		\label{property:number-trees-small}
		\item Let $T \in \calT_\tau$ and $\tilde T$ be the tree obtained from $T$ by un-contracting the nodes of $T$ to their correspondent black components. Then $\tilde T$ contains at most $8\ell$ vertices. \label{property:size-small}
	\end{properties}
\end{lemma}
\begin{proof}
	Property~\ref{property:every-node-appears-once} simply follows from the construction of the $\calT_\tau$. 
	Focus on a tree $T = (P, E_T) \in \calT_\tau$ rooted at $r$ that is not the last tree added to $\calT_\tau$. The total weight of all nodes in $P \setminus r$ is $y(U(P \setminus r)) \in [\ell, 2\ell]$.   Let the ``contributing set'' of the tree $T$ be $P \setminus r$. The total weight all nodes in the last tree added to $\calT_\tau$ is at least $\ell$, since the root node $r_\tau$ has weight at least $\ell$. Let the ``contributing set'' of the last tree $T$ be $P$. Since the contributing sets are disjoint and each contributing set has weight at least $\ell$,  the number of trees in $\calT_\tau$ is at most $1/\ell$ times the total weight of $\tau$. We proved Property~\ref{property:number-trees-small}.
	
	Focus on a tree $T = (P, E_T) \in \calT_\tau$ such that the root of $T$ is not $r_\tau$. If we uncontract all nodes in $P$ to their respective black components, then the number of vertices in the resulting tree is at most $6\ell$. This holds since the total weight of all nodes in $P$ is at most $3\ell$ and each vertex $v$ has weight $y(U_v) \geq 1/2$.  Now suppose the root of $T$ is $r_\tau$. The root node $r_\tau$ either has weight at most $2\ell$, or $J_{r_\tau}$ is a singleton. In either case, we can bound the number of vertices in $\tilde T$ by $(2\ell + 2\ell)/(1/2) = 8\ell$.  So Property~\ref{property:size-small} holds.
\end{proof}

The sub-trees in $\calT_\tau$ for all trees $\tau \in \upsilon$ gave us a partition of $R$.  For every tree $\tau \in \Upsilon$ and every $T = (P, E_T) \in \calT_\tau$ with root $r$, there is a group $J_{P \setminus r}$ in the partition.  Also, for any root $r_\tau$ of $\tau \in \Upsilon$, there is a group $J_{r_\tau}$ in the partition. Our algorithm handles each group separately. 
 
\paragraph{Pre-opening facilities and pre-assigning clients} 
In order to bound the moving cost within each group, we need to handled the so-called ``concentrated sets''.  Before defining concentrated sets, we need to define an important function of sets of representatives:
\begin{definition}
	Let $\pi(J) = \sum_{j \in C}x_{U(J),j} (1-x_{U(J), j})$, for every $J \subseteq R$. 
\end{definition}
The next lemma shows the importance of the function $\pi(J)$.
\begin{lemma}
	\label{lemma:x-j-times-one-minus-x-j-small}
	Given any non-trivial subset $J \subseteq R$ of representatives, we have 
	\begin{equation*}
		\setlength{\abovedisplayskip}{3pt}\setlength{\belowdisplayskip}{3pt}
		d(J, R \setminus J) \pi(J) \leq 4D(U_J) + 10D'(U_J).
	\end{equation*}
\end{lemma}

Notice that in an isolated group $B = U_J$ in the gap instance, each client is either completely served by $B$, or completely served by $F \setminus B$. Thus $\pi(J)$ is $0$. So the inequality holds no matter how big $d(J, R \setminus J)$ is.  In some sense, $\pi(J)$ measures how many clients are both served by $U_J$ and $F \setminus U_J$.  The bigger $\pi(J)$, the smaller $d(J, R \setminus J)$ is. 

\begin{definition}
	A set $J \subseteq R$ of representatives is said to be \emph{concentrated} if $\pi(J) \leq x_{U(J), C} /\ell_2$.
\end{definition}

Recall that $x_{U(J),C}$ is the total demand in $U_J$ after all demands are moved to the representatives using the algorithm in Section~\ref{sec:basic-rounding} and $\ell_2 = \Theta(\ell\log\ell)$ is a large enough number. Thus, according to Lemma~\ref{lemma:x-j-times-one-minus-x-j-small}, if $J$ is not concentrated, we can use $D(U_J) + D'(U_J)$ to charge the cost for moving all the $x_{U(J), C}$ units of demand out of $J$, provided that the moving distance is not too big compared to $d(J, R \setminus J)$. If $J$ is concentrated, the amount of demand that is moved out of $J$ must be comparable to $\pi(J)$.  To achieve this goal, we will use the following lemma to pre-assign some clients $C'$ so that the remaining demands $x_{U_J, C \setminus C'}$ inside $U_J$ is comparable to $\pi(J)$.

\begin{lemma}
	\label{lemma:dealing-with-concentrated-sets} If $\ell_2 = O(\ell \log \ell)$ is large enough then the following is true.
	Let $J \subseteq R$ be a concentrated set and $B = U_J$ satisfies $y_B \leq 2\ell$. Moreover, Constraints~\eqref{CLPC:add-to-one} to~\eqref{CLPC:more-than-ell-facilities} are satisfied for $B$.  Then, we can pre-open a set $S \subseteq B$ of 
	facilities and pre-assign a set $C' \subseteq C$ of clients to $S$ such that
	\begin{properties}{lemma:dealing-with-concentrated-sets}
		\item each facility $i \in S$ is pre-assigned at most $u_i$ clients; \label{property:concentrated-capacity}
		\item $x_{B, C\setminus C'} \leq \ell_2 \pi(J)$;\label{property:concentrated-satisfy-enough-demands}
		\item $\frac{x_{B, C \setminus C'}}{x_{B, C}}{y_B} + |S| \leq \Big(1+\frac1\ell\Big)y_B$; 
		\label{property:concentrated-not-too-many-facilities}
		\item the cost for the pre-assignment is at most $\ell_2 D_B$. 
		\label{property:concentrated-distance-small}
	\end{properties}
\end{lemma}

Property~\ref{property:concentrated-satisfy-enough-demands} says that the remaining demand in $B = U_J$ is small. Property~\ref{property:concentrated-distance-small} says that the cost for the pre-assignment can be charged locally using $D_B$.  Property~\ref{property:concentrated-not-too-many-facilities} says that even after we pre-opened the facilities in $S$, we can still afford to open $\frac{x_{B, C \setminus C'}}{x_{B, C}}y_B$ facilities. Notice that if $J$ is not concentrated, we can simply take $S = \emptyset $ and $C = \emptyset$ to satisfy all the properties. 

We defer the formal proof of Lemma~\ref{lemma:dealing-with-concentrated-sets} to Appendix~\ref{appendix:proofs} and give the key ideas here, assuming $\pi(J) = 0$ and $z^B_\bot = 0$. The fractional vector on variables $\set{x_{i,j}}_{i \in B, j \in C}$ and $\set{y_i}_{i \in B}$ can be expressed as a convex combination of valid integral vectors.  That is,  we can randomly open a set $S \subseteq B$ of facilities and connect some facilities $C'$ to facilities in $S$ such that: (i) each facility $i$ is open with probability $y_i$; (ii) each client $j$ is connected with probability $x_{B,j}$ and (iii) each facility $i$ is connected by at most $u_i$ clients. $\pi(J) = 0$ implies that every client $j$ has either $x_{B, j} = 0$ or $x_{B, j} = 1$.  If $x_{B, j} = 1$, $j \in C'$ always holds. Thus, we always have $x_{B, C \setminus C'} = 0$, which implies Properties~\ref{property:concentrated-satisfy-enough-demands} and Property~\ref{property:concentrated-not-too-many-facilities}. The expected assignment cost will be at most $D_B$. We condition on the event that $|S| \leq (1+1/\ell)y_B$, which happens with probability at least $\Omega(1/\ell)$. So, under this condition, the expected connection cost will be $O(\ell) D_B$, implying Property~\ref{property:concentrated-distance-small}.  When $\pi(J) > 0$, we use a smooth version of this proof.  As we only conditioned on the event that $|S| \leq (1+1/\ell)y_B \leq \ell_1$, $z^B_\bot  > 0$ is not an issue. In the proof, we group the sets in $\calS$ into $O(\log \ell)$ groups; this is the reason that we require $\ell_2 = \Theta(\ell \log \ell)$. 

With Lemma~\ref{lemma:dealing-with-concentrated-sets}, we can now pre-open some facilities, and pre-assign some clients to these pre-opened facilities.   We handle each tree $\tau \in \Upsilon$ separately.  For every node $p$ in $\tau$ other than the root $r_\tau$, the weight of $p$ is at most $\ell$. For the root $r_\tau$, either the weight of $r_\tau$ is at most $2\ell$, or $J_{r_\tau}$ contains only one vertex. So, we apply the following procedure to $p$ if $p$ is not the root, or if $p$ is the root and the weight of $p$ is at most $2\ell$.  If $J_p$ is concentrated, we check if Constraints~\eqref{CLPC:add-to-one} to \eqref{CLPC:more-than-ell-facilities} are satisfied for $B = U_p$. If not, we return the set $B = U_p$.  Otherwise, we apply Lemma~\ref{lemma:dealing-with-concentrated-sets} to the set $J_p$ to pre-open some facilities and pre-assign some clients.  Notice that the total cost for pre-assignment is at most $\ell_2 D_F = \ell_2\LP$, as the sets $J$ for which we apply Lemma~\ref{lemma:dealing-with-concentrated-sets} are disjoint.

As the pre-assignment processes for all nodes $p$ are done independently, it is possible that a client is pre-assigned more than once. This is not an issue as we only over-estimated the cost for the pre-assignment. Let $\tC$ be the set of clients that are never pre-assigned.

\paragraph{Moving Demands} We only need to focus on clients in $\tC$ now. Each client $j \in \tC$ initially has one unit of demand.  We need to move all demands to facilities. First, for every $j \in \tC$ and $i \in F$, we move $x_{i, j}$ units of demand from $j$ to $i$. The moving cost is $\sum_{j \in \tC}\dav(j) \leq \LP$.
After this step, each facility $i \in F$ has $x_{i, \tC}$ units of demand.

Let us focus on a tree $\tau$ in the forest $\Upsilon$ and a tree $T = (P, E) \in \calT_\tau$ with root node $r$. $P$ is the set of nodes in $T$ and $E$ is the set of grey edges.    Recall that $U_{P \setminus r} = \union_{p \in P \setminus r}U_p$.  We move the demands within $U_{P \setminus r}$. The moving process is simple. Let $v^* \in J_r$ be an arbitrary representative in $J_r$. We first move all demands in $U_{P \setminus r}$ to $v^*$, then we move the demand in $v^*$ to $U_{P \setminus r}$ according to some distribution.  


In order to find the distribution, we shall use  $\alpha_i$ to denote the amount of demand we shall move to facility $i$,  for any $i \in U_{P \setminus r}$. Let $\alpha(S) :=  \sum_{i \in S} \alpha_i$ for every $S \subseteq U_{P \setminus r}$.  Let $t$ be the number of pre-opened facilities in $U_{P \setminus r}$. As we are considering the soft-capacitated case, we can open the facility $i$ even if $i$ is pre-opened.   We solve the following LP to obtain $\{\alpha_i\}_{i \in U_{P \setminus r}}$:
\begin{align}
\min \qquad \sum_{i \in U_{P \setminus r}}\alpha_i d(i, v^*)\qquad \text{ s.t.} \hspace{0.25\textwidth}\label{equ:polytopeobj}\\
		\alpha_i \in [0, u_i], \forall i \in U_{P \setminus r}; \qquad 
		\alpha(U_{P \setminus r}) = x_{U_{P \setminus r}, \tilde C}; \qquad
		\sum_{i \in U_{P \setminus r}}\frac{\alpha_i}{u_i} + t  \leq \big(1+\frac1\ell\big)y(U_{P \setminus r}). \nonumber
\end{align}

The objective function of LP\eqref{equ:polytopeobj} is the moving cost.  The first constraint requires that we move at most  $u_i$ units of demand to $i$, the second constraint says the total demand is $x_{U_{P \setminus r}, \tC}$, and the last constraint bounds the total number of open facilities, including the $t$ pre-opened facilities.  We give a valid solution $\{\tilde \alpha_i\}_{i \in U_{P \setminus r}}$ for the LP: for every $p \in P \setminus r$ and $i \in U_p$, we let $\tilde \alpha_i = \frac{x_{U_p, \tC}}{x_{U_p, C}}x_{i,C}$. 
\begin{claim}
	\label{claim:solution-in-polytope}
	$\{\tilde \alpha_i\}_{i \in U_{P \setminus r}}$ satisfies the constraints of LP\eqref{equ:polytopeobj}.
\end{claim}
We find a vertex solution $(\alpha^*_i)_{i \in U_{P \setminus r}}$ of LP\eqref{equ:polytopeobj}. So, for all but at most $2$ facilities $i \in U_{P \setminus r}$, we have $\alpha^*_i \in \set{0, u_i}$. $\{\alpha^*_i\}_{i \in U_{P \setminus r}}$ satisfies the constraints of LP\eqref{equ:polytopeobj}, and $\sum_{i \in U_{P \setminus r}}\alpha^*_i d(i, v^*) \leq \sum_{i \in U_{P \setminus r}}\tilde \alpha_i d(i, v^*)$.

We open a facility at $i$ if $\alpha^*_i > 0$. The total number of open facilities, including the pre-opened facilities, is at most $\sum_{i \in U_{P \setminus r}}\ceil{\alpha^*_i/u_i} + t \leq \sum_{i \in U_{P \setminus r}}\alpha^*_i/u_i + 2 + t \leq (1+1/\ell)y(U_{P \setminus r}) + 2$.

The total moving cost is at most 
	$\sum_{i \in U_{P \setminus r}} (x_{i, \tC} + \alpha^*_i)d(i, v^*) \leq \sum_{i \in U_{P \setminus r}} (x_{i, \tC} + \tilde \alpha_i)d(i, v^*)$.
Focus on a node $p \in P \setminus r, v \in J_p$ and $i \in U_v$. We have
\begin{align*}
	(x_{i, \tC} + \tilde \alpha_i)d(i, v^*) &\leq (x_{i, \tC} + \tilde \alpha_i)(d(i, v) + d(v, v^*)) \leq 2x_{i, C} d(i, v) + \Big(x_{i, \tC} + \frac{x_{U_p, \tC}}{x_{U_p, C}}x_{i, C}\Big)d(v, v^*)\\
	&\leq 2 x_{i, C}d(i, v) + 8\ell \Big(x_{i, \tC} + \frac{x_{U_p, \tC}}{x_{U_p, C}}x_{i, C}\Big) d(J_p, R \setminus J_p).
\end{align*}
The first inequality is by triangle inequality. The second inequality used $x_{i, \tC} \leq x_{i, C}$ and $\tilde \alpha_i \leq x_{i, C}$. By Properties~\ref{property:lengths-decrease}, \ref{property:grey-edge-length} and \ref{property:black-edge-length}, all edges in the path from $v$ to $v^*$ in $\MST$ have length at most $d(J_p, R \setminus J_p)$. By Property~\ref{property:size-small}, we have $d(v, v^*) \leq 8\ell d(J_p, R \setminus J_p)$, implying the third inequality.

Fix $p \in P \setminus r$. We sum up the above inequality over all $v \in J_p$ and $i \in U_v$. The first term becomes
	$$\textstyle 2\sum_{v \in J_p, i \in U_v}x_{i,C}d(i, v) \leq 2\sum_{v \in J_p}(D(U_v) + 4D'(U_v)) = O(1)(D(U_p) + D'(U_p)),$$
by Lemma~\ref{lemma:moving-to-representatives}. The second term becomes 
\begin{align*}
	8\ell \sum_{i \in U_p}\Big(x_{i, \tC} + \frac{x_{U_p, \tC}}{x_{U_p, C}}x_{i, C}\Big) d(J_p, R \setminus J_p) = O(\ell) x_{U_p, \tC}d(J_p, R \setminus J_p).
\end{align*}
If $J_p$ is not concentrated, then  $x_{U_p, \tC} \leq x_{U_p, C} \leq \ell_2 \pi(J_p)$, by the definition of concentrated sets. If $J_p$ is concentrated, then $x_{U_p, \tC} \leq \ell_2 \pi(J_p)$ by Property~\ref{property:concentrated-satisfy-enough-demands}. So,  $x_{U_p, \tC} \leq \ell_2 \pi(J_p)$ always holds. The above quantity is at most  $O(\ell\ell_2)\pi(J_p)d(J_p, R \setminus J_p) \leq O(\ell \ell_2)(D(U_p) + D'(U_p))$ by Lemma~\ref{lemma:x-j-times-one-minus-x-j-small}.

So, we have 
	$\sum_{i \in U_p}(x_{i, \tC} + \tilde \alpha_i)d(i, v^*) \leq O(\ell\ell_2)(D(U_p) + D'(U_p))$.
Summing up over all $p \in P \setminus r$, we obtain
$\sum_{i \in U_{P \setminus r}}(x_{i, \tC} + \tilde \alpha_i)d(i, v^*) \leq O(\ell\ell_2)(D(U_{P \setminus r}) + D'(U_{P \setminus r}))$.


For a tree $\tau \in \Upsilon$ with root $r = r_\tau$, we also need to move demands within the facilities in $U_r$. Observe that all black edges in $J_r \choose 2$ have length at most $d(J_r, R \setminus J_r)$ and $J_r$ has at most $4\ell$ vertices. Use exactly the same argument as above we can bound the number of open facilities (including the possible pre-opened facilities) in $U_r$ by $(1+\frac1\ell)y(U_r) + 2$ and the moving cost by $O(\ell\ell_2)(D(U_r) + D'(U_r)$. That is, we let $v^*$ be an arbitrary vertex in $J_r$.   We move all demands in $U_r$ first to $v^*$; then we move the demand in $v^*$ to facilities $U_r$, according to the distribution $\{\alpha^*_i\}_{i \in U_r}$ obtained by solving LP\eqref{equ:polytopeobj}, with $U_{P \setminus r}$ replaced by $U_r$.  All the above inequalities hold with $P \setminus r$ replaced with $r$. 

Taking all trees $T \in \calT_\tau$ into consideration, the total number of open facilities in $U_{P_\tau}$ is at most $(1+1/\ell)y(U_{P_\tau}) + 2(|\calT_\tau| + 1)$, where $P_\tau$ is the set of nodes in $\tau$.  By Property~\ref{property:number-trees-small}, this is at most $(1+1/\ell)y(U_{P_\tau}) + 2(y(U_{P_\tau})/\ell + 1) \leq (1+5/\ell)y(U_{P_\tau})$, since $y(U_{P_\tau}) \geq \ell$.  The total moving cost is at most $O(\ell\ell_2)D(U_{P_\tau} + D'(U_{P_\tau}))$.
Taking all trees $\tau \in \Upsilon$ into consideration, the total number of open facilities is at most $(1+5/\ell)y_F \leq (1+5/\ell)k$.  The total moving cost, is $O(\ell\ell_2)(D(F) + D'(F)) = O(\ell\ell_2)\LP$. Setting $\ell = \ceil{5/\eps}$, we obtain our $\Big(1+\eps, O\big(\frac{1}{\eps^2}\log\frac1\eps\big)\Big)$-approximation for \CKM. Due to the pre-opening, a facility may be opened twice in our solution. 

\section{Discussion}
\label{sec:discussion}

In this paper, we proposed a $\Big(1+\eps, O\big(\frac{1}{\eps^2}\log\frac1\eps\big)\Big)$-approximation for \CKM.  We introduced a novel configuration LP for the problem which has small integrality gap with $(1+\eps)k$ open facilities. 

There are some open problems related to our result. Our algorithm opens at most 2 copies of each facility.  Can we reduce the number of copies for each facility to 1 so that we can extend the result to hard \CKM? Can we get constant approximation for \CKM with $(1+\eps)$-violation on the capacity constraints?  Finally, can we get a true constant approximation for \CKM? The problem is open even for a very special case: all facilities have the same capacity $u$, the number of clients is exactly $n = ku$, and $F = C$(which can be assumed w.l.o.g by \cite{Li15}).


\bibliographystyle{plain}
\bibliography{reflist}

\begin{thebibliography}{10}

\bibitem{ABG14}
Karen Aardal, Pieter van~den Berg, Dion Gijswijt, and Shanfei Li.
\newblock Approximation algorithms for hard capacitated $k$-facility location
  problems.
\newblock {\em CoRR}, abs/1311.4759v4, 2014.

\bibitem{AAB10}
Ankit Aggarwal, L.~Anand, Manisha Bansal, Naveen Garg, Neelima Gupta, Shubham
  Gupta, and Surabhi Jain.
\newblock A 3-approximation for facility location with uniform capacities.
\newblock In {\em Proceedings of the 14th International Conference on Integer
  Programming and Combinatorial Optimization}, IPCO'10, pages 149--162, Berlin,
  Heidelberg, 2010. Springer-Verlag.

\bibitem{ASS14}
Hyung-Chan An, Mohit Singh, and Ola Svensson.
\newblock {LP}-based algorithms for capacitated facility location.
\newblock In {\em Proceedings of the 55th Annual IEEE Symposium on Foundations
  of Computer Science, FOCS 2014}.

\bibitem{AGK01}
V.~Arya, N.~Garg, R.~Khandekar, A.~Meyerson, K.~Munagala, and V.~Pandit.
\newblock Local search heuristic for k-median and facility location problems.
\newblock In {\em Proceedings of the thirty-third annual ACM symposium on
  Theory of computing}, STOC '01, pages 21--29, New York, NY, USA, 2001. ACM.

\bibitem{BGG12}
Manisha Bansal, Naveen Garg, and Neelima Gupta.
\newblock A 5-approximation for capacitated facility location.
\newblock In {\em Proceedings of the 20th Annual European Conference on
  Algorithms}, ESA'12, pages 133--144, Berlin, Heidelberg, 2012.
  Springer-Verlag.

\bibitem{Byr07}
J.~Byrka.
\newblock An optimal bifactor approximation algorithm for the metric
  uncapacitated facility location problem.
\newblock In {\em APPROX '07/RANDOM '07: Proceedings of the 10th International
  Workshop on Approximation and the 11th International Workshop on
  Randomization, and Combinatorial Optimization. Algorithms and Techniques},
  pages 29--43, Berlin, Heidelberg, 2007. Springer-Verlag.

\bibitem{BFR15}
Jaroslaw Byrka, Krzysztof Fleszar, Bartosz Rybicki, and Joachim Spoerhase.
\newblock Bi-factor approximation algorithms for hard capacitated k-median
  problems.
\newblock In {\em Proceedings of the 26th Annual ACM-SIAM Symposium on Discrete
  Algorithms (SODA 2015)}.

\bibitem{BPR15}
Jaroslaw Byrka, Thomas Pensyl, Bartosz Rybicki, Aravind Srinivasan, and Khoa
  Trinh.
\newblock An improved approximation for $k$-median, and positive correlation in
  budgeted optimization.
\newblock In {\em Proceedings of the 26th Annual ACM-SIAM Symposium on Discrete
  Algorithms (SODA 2015)}.

\bibitem{Car00}
Robert~D. Carr, Lisa~K. Fleischer, Vitus~J. Leung, and Cynthia~A. Phillips.
\newblock Strengthening integrality gaps for capacitated network design and
  covering problems.
\newblock In {\em Proceedings of the Eleventh Annual ACM-SIAM Symposium on
  Discrete Algorithms}, SODA '00, pages 106--115, Philadelphia, PA, USA, 2000.
  Society for Industrial and Applied Mathematics.

\bibitem{CG99}
M.~Charikar and S.~Guha.
\newblock Improved combinatorial algorithms for the facility location and
  k-median problems.
\newblock In {\em In Proceedings of the 40th Annual IEEE Symposium on
  Foundations of Computer Science}, pages 378--388, 1999.

\bibitem{CGT99}
M.~Charikar, S.~Guha, E.~Tardos, and D.~B. Shmoys.
\newblock A constant-factor approximation algorithm for the k-median problem
  (extended abstract).
\newblock In {\em Proceedings of the thirty-first annual ACM symposium on
  Theory of computing}, STOC '99, pages 1--10, New York, NY, USA, 1999. ACM.

\bibitem{CS04}
F.~A. Chudak and D.~B. Shmoys.
\newblock Improved approximation algorithms for the uncapacitated facility
  location problem.
\newblock {\em SIAM J. Comput.}, 33(1):1--25, 2004.

\bibitem{CW05}
Fabian~A. Chudak and David~P. Williamson.
\newblock Improved approximation algorithms for capacitated facility location
  problems.
\newblock {\em Math. Program.}, 102(2):207--222, March 2005.

\bibitem{CR05}
Julia Chuzhoy and Yuval Rabani.
\newblock Approximating k-median with non-uniform capacities.
\newblock In {\em SODA ’05}, pages 952--958, 2005.

\bibitem{FMP12}
Samuel Fiorini, Serge Massar, Sebastian Pokutta, Hans~Raj Tiwary, and Ronald
  de~Wolf.
\newblock Linear vs. semidefinite extended formulations: Exponential separation
  and strong lower bounds.
\newblock In {\em Proceedings of the Forty-fourth Annual ACM Symposium on
  Theory of Computing}, STOC '12, pages 95--106, New York, NY, USA, 2012. ACM.

\bibitem{GK98}
S~Guha and S~Khuller.
\newblock Greedy strikes back: Improved facility location algorithms.
\newblock In {\em Journal of Algorithms}, pages 649--657, 1998.

\bibitem{JMM03}
K.~Jain, M.~Mahdian, E.~Markakis, A.~Saberi, and V.~V. Vazirani.
\newblock Greedy facility location algorithms analyzed using dual fitting with
  factor-revealing {LP}.
\newblock {\em J. ACM}, 50:795--824, November 2003.

\bibitem{JMS02}
K.~Jain, M.~Mahdian, and A.~Saberi.
\newblock A new greedy approach for facility location problems.
\newblock In {\em Proceedings of the thiry-fourth annual ACM symposium on
  Theory of computing}, STOC '02, pages 731--740, New York, NY, USA, 2002. ACM.

\bibitem{JV01}
K~Jain and V.~V. Vazirani.
\newblock Approximation algorithms for metric facility location and k-median
  problems using the primal-dual schema and {L}agrangian relaxation.
\newblock {\em J. ACM}, 48(2):274--296, 2001.

\bibitem{KPR98}
M.~R. Korupolu, C.~G. Plaxton, and R.~Rajaraman.
\newblock Analysis of a local search heuristic for facility location problems.
\newblock In {\em Proceedings of the ninth annual ACM-SIAM symposium on
  Discrete algorithms}, SODA '98, pages 1--10, Philadelphia, PA, USA, 1998.
  Society for Industrial and Applied Mathematics.

\bibitem{Li14}
Shanfei Li.
\newblock An improved approximation algorithm for the hard uniform capacitated
  k-median problem.
\newblock In {\em APPROX '14/RANDOM '14: Proceedings of the 17th International
  Workshop on Combinatorial Optimization Problems and the 18th International
  Workshop on Randomization and Computation}, APPROX '14/RANDOM '14, 2014.

\bibitem{Li15}
Shi Li.
\newblock On uniform capacitated $k$-median beyond the natural {LP} relaxation.
\newblock In {\em Proceedings of the 26th Annual ACM-SIAM Symposium on Discrete
  Algorithms (SODA 2015)}.

\bibitem{Li11}
Shi Li.
\newblock A 1.488 approximation algorithm for the uncapacitated facility
  location problem.
\newblock In {\em Automata, Languages and Programming - 38th International
  Colloquium (ICALP)}, pages 77--88, 2011.

\bibitem{LS13}
Shi Li and Ola Svensson.
\newblock Approximating k-median via pseudo-approximation.
\newblock In {\em Proceedings of the Forty-fifth Annual ACM Symposium on Theory
  of Computing}, STOC '13, pages 901--910, New York, NY, USA, 2013. ACM.

\bibitem{LV92B}
J.~Lin and J.~S. Vitter.
\newblock Approximation algorithms for geometric median problems.
\newblock {\em Inf. Process. Lett.}, 44:245--249, December 1992.

\bibitem{MYZ06}
M.~Mahdian, Y.~Ye, and J.~Zhang.
\newblock Approximation algorithms for metric facility location problems.
\newblock {\em SIAM J. Comput.}, 36(2):411--432, 2006.

\bibitem{STA97}
D.~B. Shmoys, E.~Tardos, and K.~Aardal.
\newblock Approximation algorithms for facility location problems (extended
  abstract).
\newblock In {\em STOC '97: Proceedings of the twenty-ninth annual ACM
  symposium on Theory of computing}, pages 265--274, New York, NY, USA, 1997.
  ACM.

\bibitem{Yan91}
Mihalis Yannakakis.
\newblock Expressing combinatorial optimization problems by linear programs.
\newblock {\em Journal of Computer and System Sciences}, 43(3):441 -- 466,
  1991.

\bibitem{ZCY05}
Jiawei Zhang, Bo~Chen, and Yinyu Ye.
\newblock A multiexchange local search algorithm for the capacitated facility
  location problem.
\newblock {\em Math. Oper. Res.}, 30(2):389--403, May 2005.

\end{thebibliography}

\appendix

\section{Omitted proofs}
\label{appendix:proofs}

\subsection{Proof of Claim~\ref{claim:representatives}}
\begin{proof}
	First consider Property~\ref{property:representatives-far-away}.  Assume $\dav(v) \leq \dav(v')$. When we add $v$ to $R$, we remove all clients $j$ satisfying $d(v, j) \leq 4\dav(j)$ from $C$. Thus, $v'$ can not be added to $R$ later. For Property~\ref{property:near-a-representative}, just consider the iteration in which $j$ is removed from $C$.  The representative $v$ added to $R$ in the iteration satisfy the property. Then consider Property~\ref{property:bundle-large}. By Property~\ref{property:representatives-far-away}, we have $B:=\set{i \in F:d(i, v) \leq 2\dav(v)} \subseteq U_v$. Since $\dav(v)=\sum_{i \in F}x_{i,v}d(i,v)$ and $\sum_{i \in F}x_{i,v} = 1$, we have $\dav(v) \geq (1-x_{B, v})2\dav(v)$, implying $y(U_v) \geq y_B \geq x_{B, v} \geq 1-\frac12$, due to Constraint~\eqref{LPC:connect-to-open}.
	
	Finally, consider Property~\ref{property:facility-to-representative}. By Property~\ref{property:near-a-representative}, there is a client $v' \in R$ such that $\dav(v') \leq \dav(j)$  and $d(v', j) \leq 4 \dav(j)$. Notice that $d(i, v) \leq d(i, v')$ since $v' \in R$ and $i$ was added to $U_v$. Thus, $d(i, v) \leq d(i, v') \leq d(i, j) + d(j, v') \leq d(i, j) + 4\dav(j)$.
\end{proof}

\subsection{Proof of Claim~\ref{claim:black-component-in-set}}
\begin{proof}
	We only need to prove that all the black edges in $J \choose 2$ are considered before all the edges in $J \times (R \setminus J)$ in the Kruskal's algorithm.  Assume otherwise.  Consider the first edge $e$ in $J \times (R \setminus J)$ we considered. Before this iteration, $J$ is not connected yet.   Then we add $e$ to the minimum spanning tree; since $J$ is a black component, $e$ is gray or white. In either case, the new group $J'$ formed by adding $e$ will have weight more than $\ell$. This implies all edges in $J' \times (R \setminus J')$ added later to the MST are not black. Moreover, $J\setminus J', J' \setminus J$ and $J \cap J'$ are all non-empty. This contradicts the fact that $J$ is a black component. 
\end{proof}

\subsection{Proof of Lemma~\ref{lemma:contracted-tree}}
\begin{proof}
		Focus on any small black component $J$ in $\MST$. By Claim~\ref{claim:black-component-in-set}, it is a group at some iteration of the Kruskal's algorithm. Consider the first iteration that we add an edge in $J \times (R \setminus J)$ to the $\MST$. This edge can not be white because $J$ is small; the edge can not be black since $J$ is a (maximum) black component. Thus, the edge must be a gray edge in $\MST$, directed from $J$ to some other black component.
		
		If for a group $J' \in \calJ$ at some iteration of Kruskal's algorithm, $J \choose 2$ contains a greay or white edge, then $J'$ is  big.   We can only add white edges between two big groups.  Let $\tilde \tau$ be the tree $\tau$ obtained by un-contract all the nodes back to the original black component. The growth of the tree $\tilde \tau$ in the Kruskal's algorithm must be as follows. First, a grey edge is added between two black components, one of them is big and the other is small.  We define the root node $r_\tau$ of $\tau$ to be the node correspondent to the big component.  At each time, we add a new small black component  $J$ to the existing tree via a grey edge with head in $J$. (During this process, white edges incident to the existing tree $\tilde \tau$ may be added.)  So, the tree $\tau$ is a rooted tree with grey edges, where all edges are directed towards the root. This proves Property~\ref{property:rooted-tree} and \ref{property:nodes-in-tree-small}. By the order we add the grey edges, we have Property~\ref{property:lengths-decrease}.  For each small black component $J$, the first grey edge in $J \times (R \setminus J)$ is the grey edge between $J$ is its parent component. Thus, we have Property~\ref{property:grey-edge-length}.  This edge is added after the black component $J$ is formed; thus we have Property~\ref{property:black-edge-length}.
\end{proof}

\subsection{Proof of Lemma~\ref{lemma:x-j-times-one-minus-x-j-small}}
\begin{proof}
Let $B = U_J$. For every $i \in B, j \in C$, we have $d(i, J) \leq d(i, j) + 4\dav(j)$ by Property~\ref{property:facility-to-representative} in Claim~\ref{claim:representatives} and the fact that $i \in U_v$ for some $v \in J$. Thus, 
\begin{align*}
d(J, R\setminus J)\pi(J) &= d(J, R \setminus J)\sum_{j \in C}x_{B, j}(1-x_{B,j}) = d(J, R \setminus J)\sum_{j \in C, i \in B, i' \in F \setminus B}x_{i,j}x_{i',j}\\
&\leq \sum_{i \in B, j \in C, i' \in F \setminus B}x_{i,j} x_{i',j} \cdot 2d(i', J) 
\leq 2\sum_{i \in B, j \in C}x_{i,j}\sum_{i' \in F}x_{i',j}\big(d(i',j) + d(j, i) + d(i, J)\big)\\
&=2\sum_{i \in B, j \in C}x_{i,j} \Big(\dav(j) + d(j, i) + d(i, J)\Big) \leq 2\sum_{i \in B, j \in C}x_{i,j} \Big(2d(i, j) + 5\dav(j) \Big) \\
&=2\sum_{i \in B}(2D_i + 5D'_i) = 4D(U_J) + 10D'(U_J).
\end{align*}
In the above sequence, the first inequality is by $d(J, R \setminus J) \leq 2 d(i', J)$ for any $i' \in F\setminus B = U_{R \setminus J}$: $d(i', R \setminus J) \leq d(i', J)$ implies $d(R \setminus J, J) \leq d(R \setminus J, i') + d(i', J) \leq 2d(i', J)$. The second inequality is by triangle inequality and the third one is by $d(i, J) \leq d(i, j) + 4\dav(j)$. All the equalities are by simple manipulations of notations. 
\end{proof}

\subsection{Pre-assignment of clients: proof of Lemma~\ref{lemma:dealing-with-concentrated-sets}}

\newcommand{\rank}{\mathsf{rank}}
\begin{proof}
 Let $Y = (1+1/\ell)y_B$.    For any $S \in \calS$ such that $|S| \leq Y$, we give $S$ a rank. If $Y - |S| < 1$, then let $\rank(S) = 0$. Otherwise, if $Y - |S| \in [2^{t-1}, 2^t)$ for some integer $t \geq 1$, then we let $\rank(S) = t$.  So, the rank of $S$ is an integer between $0$ and $\delta - 1 := \floor{\log Y} + 1 = O(\log \ell)$. 


We take the assignment of $z$ variables that satisfies Constraints~\eqref{CLPC:add-to-one} to \eqref{CLPC:more-than-ell-facilities} for $B = U_J$.  We have $\sum_{S \in \calS}z^B_S |S| + z^B_\bot \ell_1 \leq y_B$ and $\sum_{S \in \tcalS}z^B_S = 1$. So,
$\sum_{S \in \calS}z^B_S(Y - |S|) + z^B_\bot (Y - \ell_1) \geq Y - y_B = y_B/\ell$. This implies $\sum_{S \in \calS:|S| < Y}z^B_S(Y - |S|) \geq y_B/\ell$ as $Y \leq 2\ell(1+1/\ell)  = \ell_1$. Define a function $f:\{0, 1, 2, \cdots,\delta \} \to \R$ as $f(0) = Y - \floor{Y}$ and $f(t) = 2^t$ if $t \geq 1$.  If a set $S$ has rank $t \geq 1$, then $Y-|S| \in [f(t)/2, f(t)]$. If a set $S$ has rank 0, then $|S| = \floor{Y}$ and $Y-|S| = Y - \floor{Y} = f(0)$. Thus, we always have $Y-|S| \in [f(\rank(S))/2, f(\rank(S))]$. 
\begin{align*}
\sum_{S \in \calS:|S| < Y}z^B_S f(\rank(S)) \geq  \sum_{S \in \calS:|S| < Y}z^B_S(Y - |S|) \geq \frac{y_B}{\ell}.
\end{align*}

So, there is an integer $t \in [0, \delta-1]$ such that
\begin{align*}
	f(t) \sum_{S \in \calS: \rank(S) = t} z^B_S   \geq \frac{y_B}{\delta\ell}.
\end{align*}

Since $\sum_{S \in \tcalS} z^B_S = 1$, the variables $\set{z^B_S}_{S \in \tcalS}$ defines a distribution over $\tcalS$.  We randomly select an element $S \in \tcalS$: $S$ is selected with probability $z^B_S$.  Let $* = \big\{S \in \calS:\rank(S) = t\}$; we are interested in the event that $S \in *$.  Let $q = \Pr[S \in *]$. Then
\begin{align*}
	q = \sum_{S \in *}z^B_S \geq \frac{y_B}{\delta\ell f(t)} \geq \Omega\big(\frac{1}{\ell \log \ell}\big), \qquad \text{as }f(t) \leq 2Y \leq 3y_B. 
\end{align*}



Now assume we have selected a set $S \in *$. For every facility $i\in S$ and every client $j \in C$, let $w_{i, j} := z^B_{S, i, j}/z^B_{S, i} = z^B_{S, i,j}/z^B_S$, by Constraint~\eqref{CLPC:i-irrelevant}.  Notice that $w_{i,j} \in [0, 1]$ by Constraint~\eqref{CLPC:non-negative}, $\sum_{i \in S}w_{i,j} \leq 1$ for every $j \in C$ by Constraint~\eqref{CLPC:j-connection-bound}, and $\sum_{j \in C}w_{i, j} \leq u_i$ for every $i \in S$ by Constraint~\eqref{CLPC:capacity}.
Thus, $w$ is a fractional matching between $S$ and $C$ such that every facility in $i \in S$ is matched at most $u_i$ times and every client in $C$ is matched at most once. By the integrality of matching, $w$ can be expressed as a convex combination of integral matchings.  We randomly choose a matching according to the distribution defined by the combination. Let $C' \subseteq C$ be the set of matched clients. So, Property~\ref{property:concentrated-capacity} holds. For a fixed $S \in *$, the probability that $j \in C'$ is exactly $\sum_{i \in S}w_{i,j}$.

We now prove that Properties~\ref{property:concentrated-satisfy-enough-demands} and~\ref{property:concentrated-distance-small} hold simultaneously with probability at least 1/3, conditioned on $S \in *$.  To bound the probability for Property~\ref{property:concentrated-satisfy-enough-demands}, we focus on any $j \in C$ and bound its probability in $C'$. 
\begin{align}
\Pr[S \in *, j \in C'] &= \sum_{S \in *}z^B_S \sum_{i \in S}z^B_{S, i, j}/{z^B_S} = \sum_{S \in *, i \in S}z^B_{S, i, j} = x_{B,j} - \sum_{S \in \tcalS \setminus *, i \in S}z^B_{S, i, j}\nonumber \\
&\geq x_{B, j} - \sum_{S\in \tcalS \setminus *}z^B_S = x_{B, j} - 1 +  q, \label{equ:pr-j-in-J-big}
\end{align}
where the third equality is by summing up Constraint~\eqref{CLPC:add-to-x} for all facilities $i \in B$ and the inequality is by Constraint~\eqref{CLPC:j-connection-bound}.  

Let $C'' := \set{j \in C:x_{B, j} \geq 1-q}$. For any $j \in C''$, we have 
\begin{align*}
\Pr[j \notin C'\ |\ S \in *] = 1-\frac{\Pr[S \in *, j \in C']}{\Pr[S \in *]} \leq 1 - \frac{x_{B,j} - 1 + q}{q} = \frac{1-x_{B,j}}{q}.
\end{align*}
\vspace*{-\abovedisplayskip}
\begin{align}
\text{Thus, }\qquad \E\big[x_{B, C \setminus C'}\ |\ S \in *\big] &= \sum_{j \in C}x_{B,j}\Pr[j \notin C'\ |\ S \in *] \leq \sum_{j \in C''}x_{B, j}\frac{1-x_{B,j}}{q} + \sum_{j \in C \setminus C''}x_{B, j} \nonumber\\ 
&\leq \frac1q\sum_{j \in C}x_{B,j} (1-x_{B,j}) = \frac{\pi(J)}{q}, \label{inequ:satisfy-enough-demand}
\end{align}
where the equality is by linearity of expectation, and the second inequality is by $1-x_{B,j} > q, \forall j \in C \setminus C''$. 

Now, focus on Property~\ref{property:concentrated-distance-small}. For every $i \in B, j \in C$, we use $j \rightarrow i$ to denote the event that $j$ is pre-assigned to $i$. Then 
$\displaystyle \Pr[S \in *, j \rightarrow i] = \sum_{S \in *}z^B_S z^B_{S, i, j} / z^B_S = \sum_{S \in *} z^B_{S, i, j} \leq x_{i,j},$
where the inequality is by Constraint~\eqref{CLPC:add-to-x}. Thus, $\Pr[j \rightarrow i\ |\ S \in *] \leq x_{i,j}/q$. Therefore, 
\begin{align*}
\E[\text{total connection cost}\ |\ S \in *] = \sum_{i \in B, j \in C} d(i,j)\Pr[j \rightarrow i | S \in *] \leq \frac1q \sum_{i \in B, j \in C}d(i,j) x_{i,j}  = D_B/q.
\end{align*}

We set $\ell_2 = \Theta(\ell\log \ell)$ to be large enough so that $3 /q \leq \ell_2$.  By Markov inequality, $$\Pr[\text{total connection cost} \geq \ell_2 D_B\ |\ S \in *] \leq 1/3.$$ 
Also, by Markov inequality and Inequality~\eqref{inequ:satisfy-enough-demand}, we have $$\Pr[x_{B, C \setminus C'} \geq 3\pi(J)/q \ |\ S \in *] \leq 1/3.$$  By the union bound, with probability at least 1/3,  Property~\ref{property:concentrated-distance-small} holds and $x_{B, C \setminus C'} \leq 3\pi(J)/q \leq \ell_2\pi(J)$, which is exactly Property~\ref{property:concentrated-satisfy-enough-demands}. Since $J$ is concentrated, this implies $x_{B, C \setminus C'} \leq 3x_{B, C}/(\ell_2 q)$. Notice that $Y-|S| \geq f(t)/2 \geq \frac{y_B}{2\ell q \delta}$. We set $\ell_2 = \Theta(\ell \log \ell)$ to be large enough so that $6\ell \delta \leq l_2$. Then, $Y- |S| \geq \frac{3y_B}{\ell_2 q}$, implying that $\frac{x_{B, C \setminus C'}}{x_{B,C}}y_B + |S| \leq \frac{3y_B}{\ell_2 q} + |S| \leq Y$. This is exactly Property~\ref{property:concentrated-not-too-many-facilities}.
\end{proof}

\subsection{Proof of Claim~\ref{claim:solution-in-polytope}}
\begin{proof}
	Clearly, $\tilde \alpha_i \in [0, u_i]$ for every $i \in U_{P \setminus r}$. For any node $p \in P \setminus r$, we have ${\tilde \alpha}(U_p) = \frac{x_{U_p, \tC}}{x_{U_p, C}}x_{U_p, C} = x_{U_p, \tC}$. Summing up the inequalities over all $p \in P \setminus r$ gives us $\tilde \alpha(U_{P \setminus r}) = x_{U_{P \setminus r}, \tilde C}$. 
	
	Focus on a node $p \in P \setminus r$. We have $\sum_{i \in U_p}\frac{\tilde \alpha_i}{u_i} = \frac{x_{U_p, \tC}}{x_{U_p, C}}\sum_{i \in U_p} \frac{x_{i, C}}{u_i} \leq \frac{x_{U_p, \tC}}{x_{U_p, C}}\sum_{i \in U_p} y_i = \frac{x_{U_p, \tC}}{x_{U_p, C}}y(U_p)$.  If $J_p$ is not concentrated, the bound is at most $y(U_p)$. If $J_p$ is concentrated, then let $t_p$ be the number of pre-opened facilities in $U_p$ using Lemma~\ref{lemma:dealing-with-concentrated-sets}. Then,
	$\sum_{i \in U_p}\frac{{\tilde \alpha}_i}{u_i} + t_p \leq \frac{x_{U_p, \tC}}{x_{U_p, C}}y(U_p) + t_p \leq \big(1+\frac1\ell\big)y(U_p).$
	Summing up the bounds over all $p \in P \setminus r$, we have $\sum_{i \in U_{P \setminus r}}\frac{\tilde \alpha_i}{u_i} + t \leq \big(1+\frac1\ell\big)y(U_{P \setminus r})$. So, all the constraints of LP\eqref{equ:polytopeobj} are satisfied.	
\end{proof}

\end{document}